\documentclass[12pt]{amsart}
\usepackage{amsmath}
\usepackage{amssymb}
\usepackage{mathscinet}
\usepackage{amsfonts,amsmath}
\usepackage{epsfig}
\usepackage{amsthm}
\usepackage{latexsym}

\newtheorem{theorem}{Theorem}[section]
\newtheorem{lemma}[theorem]{Lemma}
\newtheorem{corollary}[theorem]{Corollary}

\newtheorem{proposition}[theorem]{Proposition}
\theoremstyle{definition}
\newtheorem{definition}[theorem]{Definition}

\theoremstyle{remark}
\newtheorem{remark}[theorem]{Remark}

\numberwithin{equation}{section}

\newcommand{\m}{\mathcal }

\newcommand{\tr}{\operatorname{Tr} }
\def\C{\mathbb C}    
\def\R{\mathbb R}    
\newcommand{\bra}[1]{\langle #1 |}    
\newcommand{\ket}[1]{| #1 \rangle}    

\begin{document}

\title{Some geometric interpretations of quantum fidelity}  

\author{Jin Li}
\address{Department of Mathematics \& Computer Science, Brandon University, Brandon, MB, Canada R7A 6A9}

\author{Rajesh Pereira}
\address{Department of Mathematics \& Statistics, University of Guelph, Guelph, ON, Canada N1G 2W1}

\author{Sarah Plosker}
\address{Department of Mathematics \& Computer Science, Brandon University, Brandon, MB, Canada R7A 6A9}

\begin{abstract}
We consider quantum fidelity between two states $\rho$ and $\sigma$, where we fix $\rho$ and allow $\sigma$ to be sent through a quantum channel. We determine the minimal fidelity where one minimizes over (a) all unital channels, (b) all mixed unitary channels, and (c) arbitrary channels. We derive results involving the minimal eigenvalue of $\rho$, which we can interpret as a convex combination coefficient. As a consequence, we give a new geometric interpretation of the minimal fidelity with respect to the closed, convex set of density matrices and with respect to the closed, convex set of quantum channels. We further investigate the geometric nature of fidelity by considering density matrices arising as normalized projections onto  subspaces; in this way, fidelity can be viewed as a geometric measure of distance between two spaces. We give a connection between fidelity and the canonical (principal) angles between the subspaces.

\end{abstract}


\maketitle

\section{Introduction}

The quantum fidelity $F(\rho, \sigma)$ is a measure of the distance between two quantum states $\rho$ and $\sigma$ that quantifies the accuracy of state transfer through a channel; the ideal case being a fidelity value of 1, which represents perfect state transfer. Physically, one begins with  initial state $\rho$ at time 0, and allows the quantum system to evolve over time. At time $t$, one measures the overlap of the two states $\rho$ and $\sigma$; this overlap decreases over time due to the evolution and perturbation of the system.

Fidelity has been considered in the context of quantum communication via  unmeasured and unmodulated spin chains which are used to transmit quantum states \cite{Bose}, and  plays a role in quantum decision tree algorithms \cite{LB14}.   Geometric interpretations of fidelity have been given in \cite{MMPZ08, MZC08} and elsewhere, however our approach and results are distinctly different from  the literature at present.

Formally, we have the following definition:
\begin{definition}\label{defn:fidelity}
Let $\rho$ and $\sigma$ be two $n\times n$ positive semidefinite  matrices. The \emph{(quantum) fidelity} between $\rho$ and $\sigma$ is
\begin{eqnarray*}
F(\rho, \sigma)&=&\tr(\sqrt{\sqrt{\rho}\sigma\sqrt{\rho}})\\
&=&\tr(\sqrt{\sqrt{\sigma}\rho\sqrt{\sigma}}).
\end{eqnarray*}
\end{definition}
Any positive semidefinite matrix has a unique positive square root, and so the quantum fidelity between two quantum states (density matrices) is well-defined and yields a non-negative real number.

The transition probability between two states, which is the square of the fidelity, was defined in \cite{U76}, although the idea stems from two earlier papers \cite{K48, B69} in a more general context. Jozsa \cite{J94} proposed four axioms that the transition probability (which he called fidelity) must satisfy (we have re-written these axioms in terms of fidelity)\footnote{The language in the literature is not consistent. Some authors take the point of view of Jozsa: transition probability = fidelity =$F^2(\rho, \sigma)$, using the notation of definition \ref{defn:fidelity}. They then call $F(\rho, \sigma)$ the square root fidelity.}

\begin{enumerate}
\item $0\leq F(\rho, \sigma)\leq 1$ with $F(\rho, \sigma)=1$ iff $\rho=\sigma$;
\item The fidelity is symmetric: $F(\rho, \sigma)=F(\sigma, \rho)$;
\item If $\rho=|\psi\rangle\langle\psi|$ is a pure state, then $F(\rho, \sigma)=\sqrt{\bra{\psi}\sigma\ket{\psi}}$;
\item \label{ax4} The fidelity is invariant under unitary transformations on the state space:
\[F(U\rho U^\dagger, U\sigma U^\dagger)=F(\rho, \sigma) \quad \textnormal{ for any unitary } U,\]
where $^\dagger$ represents complex conjugate transposition. 
\end{enumerate}

Note that, if both  $\rho=|\psi\rangle\langle\psi|$ and $\sigma=|\phi\rangle\langle\phi|$ are pure states, then the quantum fidelity $F(\rho, \sigma)$ reduces to $|\langle \psi|\phi\rangle|$. Supposing $\ket{\psi}$ and $\ket{\phi}$ are unit vectors in $\R^n$ (rather than in $\C^n$), we can write $\ket{\psi}=(\sqrt{p_1}, \dots, \sqrt{p_n})$ and  $\ket{\phi}=(\sqrt{q_1}, \dots, \sqrt{q_n})$, where $p=\{p_j\}$ and $q=\{q_j\}$ are two probability distributions. This then yields the classical fidelity
 between the two probability distributions $p=\{p_j\}$ and $q=\{q_j\}$, which  is defined as $F(p, q)=\sum_j\sqrt{p_jq_j}$. The term classical fidelity is used in quantum information theory; outside of QIT, classical fidelity is referred to as the \emph{Bhattacharyya coefficient}.

The quantum fidelity between unitary orbits of two density matrices $\rho$ and $\sigma$ is $F(V\rho V^\dagger, W\rho W^\dagger)$ for unitary $V, W$, which, in light of axiom (\ref{ax4}), reduces to $F(\rho, U\sigma U^\dagger)$ for unitary $U$.
The   maximum and minimum quantum fidelity between the unitary orbits of $\rho$ and $\sigma$ have been characterized  as follows:
\begin{theorem}\cite{MMPZ08, ZF14}\label{thm:ZF}
The quantum fidelity between unitary orbits $\m U_\rho$ and $\m U_\sigma$ satisfies the following relations:
\begin{eqnarray*}
\max_{U\in U(H_d)}F(\rho, U\sigma U^\dagger)&=&F(\lambda^\downarrow(\rho), \lambda^{\downarrow}(\sigma))\\
\min_{U\in U(H_d)}F(\rho, U\sigma U^\dagger)&=&F(\lambda^\downarrow(\rho), \lambda^{\uparrow}(\sigma)),
\end{eqnarray*}
where $U(H_d)$ is the set of all $d\times d$ unitary matrices on a $d$-dimensional Hilbert space $H_d$ and $\lambda^{\downarrow}(\rho)$ (respectively, $\lambda^{\uparrow}(\rho)$), is the vector of eigenvalues of $\rho$, listed in non-increasing (respectively, non-decreasing) order, including multiplicities.
\end{theorem}

The quantum fidelity $F(\rho, \sigma)$  was originally found to satisfy the bounds of theorem \ref{thm:ZF} in \cite{MMPZ08}, although the result found in \cite{MMPZ08} was formulated in terms of the closely related Bures distance.

For practical purposes, one wishes to maximize  $F(\rho, \sigma)$; however, it is also useful to consider minimal fidelity, which represents the worst-case scenario of quantum information state transfer.

We generalize  theorem \ref{thm:ZF} by characterizing the following minimum quantum fidelities:
\begin{enumerate}
\item the minimum quantum fidelity $F(\rho, \Phi(\sigma))$, where $\Phi$ is  any quantum channel (completely positive, trace-preserving, linear map),
\item  the minimum quantum fidelity $F(\rho, \Phi(\sigma))$, where $\Phi$ is any unital channel (a quantum channel satisfying $\Phi(I)=I$), and
\item  the minimum quantum fidelity $F(\rho, \Phi(\sigma))$, where $\Phi$ is any mixed unitary channel (a  quantum channel of the form $\Phi(\rho)=\sum_ip_iU_i\rho U_i^\dagger$, where $U_i$ are unitaries and $p_i$ form a probability distribution).
\end{enumerate}
In our derivations, we take the point of view that $\rho$ is fixed (given) and  $\sigma$ is sent through the channel $\Phi$.
Our motivation is  the case where  one has access to the output $\Phi(\sigma)$ of the state $\sigma$ after it has been sent through a channel $\Phi$, but one does not have  direct access to $\sigma$. Thus, it is of interest to see   how far away  $\rho$ and $\sigma$ can become through the use of the channel $\Phi$.

 We will show (in corollary \ref{cor:min}) that if $\Phi$ is a unital channel, then the quantity $\min_{\Phi}F(\rho, \Phi(\sigma))$, where the minimum is taken over all quantum channels $\Phi$, reduces to $F(\lambda^\downarrow(\rho), \lambda^{\uparrow}(\sigma))=\min_{U\in U(H_d)}F(\rho, U\sigma U^\dagger)$ from theorem \ref{thm:ZF}.

Our methods make extensive use of majorization of vectors of eigenvalues. It is interesting to note that majorization is also used to characterize the more specialized situation when one entangled state can be transformed into another through the use of quantum operations described by local operations and classical communication (LOCC)  \cite{Nie99}.

The paper is organized as follows.  In section \ref{sec:qf} we review majorization, the main tool used in proving the results of this section, and we derive the minimum $F(\rho, \Phi(\sigma))$ where $\Phi$ is a quantum channel, under the restrictions listed above. In section \ref{sec:geom}, we give a geometric interpretation of our results.  In section \ref{sec:proj} we continue our geometric approach, this time focusing on density matrices arising as normalized projections onto  subspaces. It appears that studying fidelity in terms of projections and subspace geometry has not been done previously. We obtain a number of new interpretations of fidelity, including theorem \ref{thm:Fcanangles}, which  links fidelity with the canonical (principal) angles between the subspaces. Section \ref{sec:dis} is devoted to a discussion on various related topics, linking our work with related results on fidelity as well as results in other areas of mathematics.

\section{Quantum fidelity when one state is sent through a quantum channel}\label{sec:qf}

\subsection{Majorization}
\begin{definition}\label{defn:maj}  Let $x=(x_{1},x_{2},...,x_{d})$ and $y=(y_{1},y_{2},...,y_{d})$ be two $d$-tuples of real numbers. We say that $(x_{1},x_{2},...,x_{d})$ is majorized  by $(y_{1},y_{2},...,y_{d})$, written $x\prec y$,   if
\begin{eqnarray*}
\sum_{j=1}^{k}x^{\downarrow}_{j}\leq \sum_{j=1}^{k}y^{\downarrow}_{j}\quad 1\leq k \leq d,
\end{eqnarray*}
with equality for $k=d$.
\end{definition}
If equality does not necessarily hold when $k=d$, we say that $x$ is \emph{sub-majorized} by $y$ and we write $x\prec_w y$, where the $w$ stands for ``weak''.

If we order the components of the vectors in \emph{non-decreasing} order, indicated by $^\uparrow$, then
 $x$ is majorized by $y$
if
\begin{eqnarray*}
\sum_{j=1}^k x^{\uparrow}_j \geq \sum_{j=1}^k y^{\uparrow}_j\quad 1\leq k \leq d,
\end{eqnarray*}
with equality when $k=d$.  This definition is equivalent to the definition of majorization given above. If equality does not necessarily hold when $k=d$, we say that $x$ is \emph{super-majorized} by $y$ and we write $x\prec^w y$.

\subsection{Minimum Fidelity}
A function $f: \R^n\rightarrow \R$ is \emph{Schur-concave} if $x\prec y\Rightarrow f(x)\geq f(y)$.

Although there are several ways of proving proposition \ref{prop:Schur-concave}, we shall prove it using Ostrowski's theorem:
\begin{theorem} \cite{Ost52}, \cite[Theorem 3.A.7]{MOA11} Let $D=\{(x_1, \dots, x_n)\,|\, x_1\geq \cdots \geq x_n\}$.
Let $\phi$ be a real-valued function defined and continuous on $D$. Then
\[
\phi(x)\leq \phi(y)\textnormal{ whenever } x\prec^w  y\textnormal{ on } D
\]
if and only if
\[
0\geq \phi_1(z)\geq \cdots \geq \phi_n(z)\,\,\forall z\textnormal{ in the interior of } D,
\]
where $\phi_i(z)=\frac{\partial \phi(z)}{\partial z_i}$.
\end{theorem}

\begin{proposition}\label{prop:Schur-concave}
Let $\{p_j \}_{j=1}^{n}$ be fixed non-negative numbers that sum to one. The function $f(q_1,...,q_n)=\sum_j\sqrt{p_j^{\uparrow}q_j^\downarrow}$ is Schur-concave.
\end{proposition}
\begin{proof}
We know that the square root function is concave, so the sum of the square root functions acting on each of the  components $q_1, \dots, q_n$ is Schur-concave.

Now, consider $\phi(z)=\sum_i\sqrt{p_i}\sqrt{z_i}$ where $z\in D$ so $z_1\geq \cdots \geq z_n$. In the minimum case, want $p_1\leq \cdots \leq p_n$. We find that $\phi_i(z)=\frac12\frac{\sqrt{p_i}}{\sqrt{z_i}}$, which increases as $i$ increases. Thus $-\phi$ satisfies Ostrowski's theorem. So if $x\prec^w y$ then $\phi(x)\geq \phi(y)$; that is, the function $f(q_1,...,q_n)=\sum_j\sqrt{p_j^{\uparrow}q_j^\downarrow}$ is Schur-concave. It follows that the absolute minimum of $\phi$ over any subset $S$ of $D$ if it exists must be at a point of $D$ which is maximal with respect to the  supermajorization order.
\end{proof}

We can use the theory of majorization to find the minimum fidelity between a fixed state $\rho$ and $\Phi(\sigma)$, the image of a second fixed state under any unital quantum channel.  We note that this result, while related to theorem \ref{thm:ZF}, is not a direct consequence of it since there exist unital quantum channels which are not the convex combination of unitary transforms \cite{LS90}. 

\begin{corollary}\label{cor:min}
If we consider unital $\Phi$, then we have
\[\min_{\Phi}F(\rho, \Phi(\sigma))=F(\lambda^\downarrow(\rho), \lambda^{\uparrow}(\sigma))\]
where the minimum is taken over all possible unital quantum channels $\Phi$.
\end{corollary}

\begin{proof}
Suppose $\rho$ and $\sigma$ are density matrices and $\Phi$ is a quantum channel. Then by Uhlmann's theorem, $\Phi(\sigma)\prec \sigma$, provided $\Phi$ is unital. Thus $f(\Phi(\sigma))\geq f(\sigma)$ for all Schur-concave functions $f$. In particular, we take the $f$ from proposition \ref{prop:Schur-concave} with $p_j=\lambda_j(\rho)$ to obtain
\begin{eqnarray*}
\sum_j\sqrt{\lambda_j^\uparrow(\rho)\lambda_j^\downarrow(\Phi(\sigma))}&\geq &\sum_j\sqrt{\lambda_j^\uparrow(\rho)\lambda_j^\downarrow(\sigma)}\\
\textnormal{i.e.\ } F(\lambda_j^\uparrow(\rho), \lambda_j^\downarrow(\Phi(\sigma)))&\geq &F(\lambda_j^\uparrow(\rho),\lambda_j^\downarrow(\sigma))\\
\textnormal{equivalently } F(\lambda_j^\downarrow(\rho), \lambda_j^\uparrow(\Phi(\sigma)))&\geq &F(\lambda_j^\downarrow(\rho),\lambda_j^\uparrow(\sigma)).
\end{eqnarray*}
Thus
$\min_{\Phi}F(\rho, \Phi(\sigma))$ (where the minimum is restricted to unital $\Phi$) is achieved precisely when $\Phi$ is the unitary transformation making the eigenvalues of $\sigma$ the same as those of $\rho$, with the eigenvalues lining up in the opposite direction, giving:
\begin{eqnarray*}
\min_{\Phi}F(\rho, \Phi(\sigma))=F(\lambda^\downarrow(\rho), \lambda^{\uparrow}(\sigma)).
\end{eqnarray*}

\end{proof}

\begin{proposition} \label{prop:minF} Let $H$ be a Hilbert space and $S(H)$ be the state space of $H$.  Let $\rho\in S(H)$ and $K$ be a subset of $S(H)$ containing all of the pure states in $S(H)$.  Then $\min_{\sigma \in K}F(\rho, \sigma)=(\lambda_{\min}(\rho))^{1/2}$, where $\lambda_{\min}(\rho)$ represents the minimal eigenvalue of $\rho$.  \end{proposition}

\begin{proof} Any mixed state $\sigma$ can be represented as a convex combination of pure states: $\sigma=\sum_ip_i\ket{\psi_i}\bra{\psi_i}$. The quantum fidelity is concave in each of its variables \cite{U76, J94}, and so $F(\rho, \sigma )\geq \sum_i p_iF(\rho , \ket{\psi_i}\bra{\psi_i})$. Hence $\sqrt{\bra{\psi_0}\rho\ket{\psi_0}}=F(\rho, \ket{\psi_0}\bra{\psi_0})\leq F(\rho, \sigma)$ for at least one of the pure states $\ket{\psi_0}$. By the Courant-Fisher theorem, we minimize $\sqrt{\bra{\psi_0}\rho\ket{\psi_0}}$ as a function of $\ket{\psi_0}$ by choosing $\ket{\psi_0}$ to be the eigenvector corresponding to the minimal eigenvalue of $\rho$ which gives us $F(\rho, \sigma)\geq \sqrt{\bra{\psi_0}\rho\ket{\psi_0}}=(\lambda_{\min}(\rho))^{1/2}$.  Since $\ket{\psi_0}\bra{\psi_0}$ is a pure state, it is in $K$ and our result follows.
 \end{proof}

As a corollary we have the following result:

\begin{corollary}\label{prop:minev}
We have \[\min_{\Phi}F(\rho, \Phi(\sigma))=(\lambda_{\min}(\rho))^{1/2}\]
where  the minimum on the left hand side of the equation is taken over all possible quantum channels $\Phi$.
\end{corollary}

\begin{proof}
Let $K=\{\Phi(\sigma):  \Phi$ is a quantum channel$\}$.
We note that $K$ contains all pure states: To see this consider $F(\rho, \Phi(\sigma))$ and take $\Phi(\cdot)=\tr(\cdot)\ket{\psi}\bra{\psi}$. The map $\Phi$ is clearly completely positive, trace-preserving, and linear, and so $\Phi$ is a quantum channel and $\ket{\psi}\bra{\psi}\in K$.  The result now follows from the previous proposition.
\end{proof}

\begin{remark}
If we choose $\Phi$ to be the quantum channel $\Phi(\cdot)=\tr(\cdot)\rho$, then this choice gives us $F(\rho,\Phi(\sigma))=F(\rho, \rho)=1$, so the maximum value of $F(\rho, \Phi(\sigma))$ is one. It is therefore trivial to find the maximum of $F(\rho, \Phi(\sigma))$ when $\Phi$ is any quantum channel. The problem becomes interesting when we restrict to unital or to mixed unitary channels, since in these special cases we no longer have complete freedom. However, we do not have results such as proposition \ref{prop:Schur-concave} and the Courant-Fisher theorem at our disposal, so finding the maximum is not a straightforward task.
\end{remark}

\begin{corollary} If we consider mixed unitary channels $\Phi$
(all channels of the form $\Phi(\cdot)=\sum_jp_jU_j(\cdot)U_j^\dagger$ where $\{p_j\}$ is a probability distribution and $U_j$ are unitaries), then
\[\min_{\Phi}F(\rho, \Phi(\sigma))=F(\lambda^\downarrow(\rho), \lambda^{\uparrow}(\sigma))\]
where the minimum is taken over all mixed unitary channels $\Phi$.
\end{corollary}
\begin{proof}

The concavity of the quantum fidelity gives us that the minimum of $F(\rho, \Phi(\sigma))$ will occur at an extreme point $\Phi$ of the set of quantum channels. Thus, if we are considering the set of mixed unitary channels, then the minimum must occur at a unitary channel: a channel of the form $\Phi(\cdot)=U(\cdot)U^\dagger$. Thus, by theorem \ref{thm:ZF}, it follows that
\[\min_{\Phi}F(\rho, \Phi(\sigma))=F(\lambda^\downarrow(\rho), \lambda^{\uparrow}(\sigma))\]
where the minimum is taken over all mixed unitary channels $\Phi$.
\end{proof}

Again we stress that the maximum value of the fidelity   could potentially occur at any point, so finding the maximum is a much more difficult matter.

\section{Geometric Interpretation of Minimum Quantum Fidelity}\label{sec:geom}
The set of all states is a compact convex set. At its center is the maximally mixed state $\frac1nI$; its boundary is made up of all singular (non-invertible) density matrices $\omega$. Any state $\rho$ can be written as a convex combination
\[
\rho=p\omega+(1-p)\left(\frac1nI\right)
\]
for some $\omega$ on the boundary, where $0\leq p\leq 1$.

Similar convex combinations have been studied in quantum information theory, and in many other fields of mathematics and computer science. For instance, the set of all channels is a compact convex set. At its center is the completely depolarizing channel $\Omega:\rho\mapsto \frac1nI$; its boundary is made up of all channels $\Psi$ whose Choi matrix $C_{\Psi}$ is singular.
Recall that the Choi matrix $C_{\Phi}$ corresponding to a channel $\Phi$ is   defined by
    \[C_\Phi= \left (I_n\otimes\Phi \right ) \left (\sum_{ij}E_{ij}\otimes E_{ij} \right ) = \sum_{ij}E_{ij}\otimes\Phi(E_{ij}), \]
    where $E_{ij}$ are the matrix units.
The Choi matrix for a channel is singular precisely when the number of Kraus operators $V_i$ in the decomposition $\Phi(\rho)=\sum_{i=1}^k V_i\rho V_i^\dagger$ minimizing $k$ is strictly less than $n^2$. (Thus most channels that arise naturally are on the boundary of the set of all quantum channels). Any channel $\Phi$ can be written as a convex combination
\begin{eqnarray}\label{eq:Phi}
\Phi=p\Psi+(1-p)\Omega
\end{eqnarray}
for some $\Psi$ on the boundary, where $0\leq p\leq 1$.

A more specific example along these lines is that of \cite{Wat09}, where  if one can write
\[
\Phi=p\Psi+(1-p)\Omega
\]
for some unital quantum channel $\Psi$, where $0\leq p\leq \frac1{d^2-1}$, then $\Phi$ is a mixed unitary channel. Note here that the set of all mixed unitary channels forms a subset of the set of all unital channels, both sets are compact convex sets, and  $\Omega$ is their common centroid.

A  result along the same vein \cite{Mar05} gives   $p$ for which
\[
\Phi=p\omega+(1-p)\left(\frac1nI\right)
\]
is a real, rank-one correlation matrix,
where $\omega$ is a real correlation matrix (a positive semi definite matrix with 1's along the diagonal).

In \cite{OP13}, the authors consider a similar convex combination problem involving $H$-unistochastic and bistochastic matrices.

Relating this to the  results herein, the value of $\lambda_{\min}(\rho)$ tells us how close $\rho$ is to the maximally mixed state, or, equivalently, how close it is to the ``extreme'' states (singular density matrices). The value of  $\lambda_{\min}(\rho)$ gets larger as $\rho$ gets closer to the maximally mixed state, and smaller as $\rho$ gets closer to the boundary of singular density matrices. In this way,  $\min_{\sigma \in K}F(\rho, \sigma)$ of proposition \ref{prop:minF} measures how far away your state $\rho$ is  from the boundary.

Similarly, in the case of quantum channels, $\Phi$ maps $\sigma$ to a density matrix with larger and larger $\lambda_{\min}$ as $p\rightarrow 0$ in equation (\ref{eq:Phi}).

\section{Fidelity, projections and subspace geometry}\label{sec:proj}

Let $S$ be an $m$-dimensional subspace of a $d$-dimensional Hilbert space $H_d$ and $P_{S}$ be the orthogonal projection onto the subspace $S$, then $\rho_{S}=\frac1m P_{S}$ is a density matrix.  Let $T$ be an $n$-dimensional subspace of $H_d$. The main goal of this section is to examine the relationship between the geometry of   two subspaces $S$ and $T$ and the quantity $F(\rho_{S},\rho_{T})$.  

\begin{remark} In the special case where $\rho_{S}$ and $\rho_{T}$ commute, we have
\begin{eqnarray*}
F(\rho_{S},\rho_{T})&=&\tr\left(\sqrt{\left(\frac1m P_{S}\right)^{1/2}\frac1nP_{T}\left(\frac1m P_{S}\right)^{1/2}}\right)\\
&=&\frac1{(mn)^{1/2}}\tr(\sqrt{P_SP_TP_S})\\
&=&\frac{\tr\sqrt{P_{S\cap T}}}{\sqrt{m n}}\\
&=&\frac{\dim(S\cap T)}{\sqrt{mn}}
\end{eqnarray*}
 so it appears that in this case the fidelity measures the proportion of overlap between the two subspaces, giving 0 when $S$ and $T$ are disjoint, and 1 when $S=T$. 
  \end{remark}
  
  We note that this result can be used to find the maximum and minimum of $F(\rho_{S},\rho_{T})$ where $S$ and $T$ range over all $m$-dimensional and $n$-dimensional subspaces of $H_d$, respectively.  By theorem \ref{thm:ZF}, both the maximum and the minimum will occur at a choice of $S$ and $T$ for which $\rho_S$ and $\rho_T$ commute.  Since $\max(m+n-d,0)\le \dim(S\cap T)\le \min(m,n)$, we get the following result.

\begin{corollary} Let $S$ and $T$ be subspaces of $H_d$ with dimension $m$ and $n$ respectively and let $\rho_S=\frac{1}{m}P_S$ and $\rho_T=\frac{1}{n}P_T$.  Then $\frac{\max(m+n-d,0)}{\sqrt{mn}}\le F(\rho_{S},\rho_{T})\le \min(\sqrt{\frac{m}{n}},\sqrt{\frac{n}{m}})$. \end{corollary}

We are interested in  properties of the fidelity of two density matrices when one or both are normalized orthogonal projections.  We have the following inequality.

\begin{proposition}\label{prop:proj} Let $S$ be an $m$-dimensional subspace of $H_d$.  Let $\rho$ be a density matrix with eigenvalues $\lambda_1\geq \lambda_2\geq ...\geq \lambda_{d-1}\geq \lambda_d$.  Then $\sum_{j=d-m+1}^{d}\sqrt{ \frac{\lambda_j}{m}}\le F(\rho, \rho_S)\le \sum_{j=1}^{m}\sqrt{\frac {\lambda_j}{m}}$.  For any fixed $\rho$, there are choices of $S$ for which $F(\rho, \rho_S)$ achieves the upper and lower bounds of the inequality respectively. \end{proposition}

\begin{proof} Let $\mu_{1}\ge \mu_{2}\ge ...\ge \mu_{m}$ be the eigenvalues of $P_S\rho P_S$ when considered as an operator on $S$.  By Cauchy's interlacing theorem, we have $\lambda_{j+d-m} \le \mu_j \le \lambda_{j}$.  Since $F(\rho, \rho_S)=\frac{1}{\sqrt{m}} tr((P_S\rho P_S)^{\frac{1}{2}})$, the result follows.  We may attain the upper and lower bound of the inequalities by choosing $S$ to be the span of the eigenvectors corresponding respectively to the $m$ largest and $m$ smallest eigenvalues of $\rho$. \end{proof}

\begin{remark}
Proposition \ref{prop:minF} follows as a corollary of proposition \ref{prop:proj}. Indeed, pure states are rank-one projections, so the dimension of the set of all pure states $K$ is $m=1$. By concavity, the minimum will occur at the boundary of the state space $S(H)$, which is precisely the pure states. We thus obtain the lower bound of proposition \ref{prop:proj}: $(\lambda_{\min}(\rho))^{1/2}$, which is precisely the result of proposition  \ref{prop:minF}.
\end{remark}

If the dimensions of $S$ and $T$ are equal, we can obtain an interesting interpretation of $F(\rho_{S},\rho_{T})$ as the average of the cosine of the canonical angles between $S$ and $T$.  Before introducing our result, we remind the reader of the definition of the canonical angles between $S$ and $T$; this concept first appears in the work of Camille Jordan in 1875 \cite{J75}.

\begin{definition} Let $S$ and $T$ be finite dimensional subspaces of a Hilbert space $\mathcal{H}$ and let $\ell=\min\{\dim (S), \dim(T)\}$.   Then the first canonical angle is the unique number $\theta_1\in [0,\frac{\pi}{2}]$ such that $\cos(\theta_1)=\max \{ \vert \langle x,y\rangle\vert: x\in S, y\in T, \Vert x\Vert=\Vert y\Vert=1\}$.  Let $x_1$ and $y_1$ be unit vectors in $S$ and $T$ respectively where the previous maximum is attained.  Then we define the second canonical angle as the unique number $\theta_2\in [0,\frac{\pi}{2}]$ such that $\cos(\theta_2)=\max \{ \vert  \langle x,y\rangle\vert: x\in S, y\in T, \Vert x\Vert=\Vert y\Vert=1, x\perp x_1, y\perp y_1\}$. Let $x_2$ and $y_2$ be the unit vectors in $S$ and $T$ respectively where the previous maximum is attained.  Now for any $k\le \ell$, $\theta_k$ is the unique number such that $\cos(\theta_k)=\max \{   \vert\langle x,y\rangle\vert: x\in S, y\in T, \Vert x\Vert=\Vert y\Vert=1, x\perp x_1, x_2, \dots x_{k-1} , y\perp y_1, y_2, \dots y_{k-1}\}$. \end{definition}

Canonical angles are also called principal angles.  We will use a characterization of the canonical angles first given in \cite{BG73}.

\begin{theorem}\label{BG} Let $S$ and $T$ be   subspaces of a Hilbert space $H$ with dimensions $m$ and $n$ respectively, and let $Q_S$ and $Q_T$ be matrices whose column vectors are the elements of orthonormal bases of $S$ and $T$ respectively.  Then the cosine of the canonical angles are   the singular values of the matrix $Q_S^\dagger Q_T$:
\[
\cos(\theta_k)=\sigma^\downarrow_k(Q_S^\dagger Q_T),
\]
for all $k=1, \dots, \ell=\min\{m, n\}$.
 \end{theorem}

We are now ready to state our main result of this section.

\begin{theorem} \label{thm:Fcanangles} Let $S$ and $T$ be   subspaces of a Hilbert space $H$ with dimensions $m$ and $n$ respectively,  with $\ell=\min\{m, n\}$. Let $P_S$ and $P_T$ be the orthogonal projections onto $S$ and $T$ respectively and let $\rho_{S}=\frac1m P_{S}$ and $\rho_{T}=\frac1n P_{T}$. The fidelity $F(\rho_{S},\rho_{T})=\frac{1}{\sqrt{mn}}\sum_{k=1}^{\ell}\cos(\theta_k)$ where $\{ \theta_{k}\}_{k=1}^{\ell}$ are the canonical angles between $S$ and $T$.\end{theorem}

Note that if $\dim(S)=\dim(T)$, then $F(\rho_{S},\rho_{T})$ is the arithmetic mean of the cosines of the canonical angles.

\begin{proof} Note that $P_S=Q_SQ_S^\dagger $ and $P_T=Q_TQ_T^\dagger $ where $Q_S$ and $Q_T$ are any matrices whose column  vectors are the elements of orthonormal bases of $S$ and $T$ respectively.  Then $\rho_S^{\frac{1}{2}}\rho_T\rho_{S}^\frac{1}{2}$ is similar to $\frac{1}{mn}P_SP_T$ which is equal to $\frac{1}{mn}Q_SQ_S^\dagger  Q_TQ_T^\dagger $ which has the same non-zero eigenvalues as $\frac{1}{mn}(Q_S^\dagger  Q_T)(Q_S^\dagger Q_T)^\dagger $ with the same multiplicities.  Therefore the non-zero eigenvalues of $(\rho_S^{\frac{1}{2}}\rho_T\rho_{S}^\frac{1}{2})^\frac{1}{2}$ are exactly the same as the non-zero singular values of $\frac{1}{\sqrt{mn}}Q_S^\dagger Q_T$.  The result now follows from theorem \ref{BG}. \end{proof}

 The Bures angle between two states $\rho$ and $\sigma$ is $\arccos (F(\rho, \sigma))$. We note here that if the density matrices are normalized orthogonal projections onto subspaces $S$ and $T$ of the same dimension, then the cosine of the Bures angle between the two states is the arithmetic mean of the cosines of the canonical angles.

\section{Discussion}\label{sec:dis}
In this section we discuss connections between the work herein and resutls found elsewhere in the mathematics and quantum information theory literature.
\subsection{Rearrangement Inequality}\label{sec:rearr}

Theorem \ref{thm:ZF} is in fact a stronger version (in the sense that it deals with non-commutative operators) of the rearrangement inequality for non-negative numbers:
\begin{eqnarray}\label{eq:rearr}
x_ny_1+\cdots x_1y_n\leq x_{\sigma(1)}y_1+\cdots +x_{\sigma(n)}y_n\leq x_1y_1+\cdots x_ny_n
\end{eqnarray}
for any choice of real numbers
\begin{eqnarray}\label{eq:xy}
 x_1\leq \cdots\leq x_n\quad\textnormal{and}\quad y_1\leq \cdots\leq y_n
\end{eqnarray}
 and for any permutation $\sigma$ of $\{1, \dots, n\}$. If we have all strict inequalities in  (\ref{eq:xy}), then the lower bound of
the inequality (\ref{eq:rearr}) is attained only for the permutation that reverses the order, i.e.\ $\sigma(i)=n-i+1$ for all $i\in \{1, \dots, n\}$  and the upper bound is attained only for the identity $\sigma(i)=i$ for all $i\in \{1, \dots, n\}$.

Indeed, consider theorem \ref{thm:ZF} under the special case where both density matrices $\rho$ and $\sigma$ are diagonal and the unitaries are permutations. With this setup, we restate theorem \ref{thm:ZF}  as
\begin{eqnarray*}
\max_{U\in U(H_d)}(F(\rho, U\sigma U^\dagger))^2&=&\sum_{i=1}^n\lambda_i(\rho)\lambda_i(\sigma),
\end{eqnarray*}
which is the upper bound of inequality (\ref{eq:rearr})  and
\begin{eqnarray*}
\min_{U\in U(H_d)}(F(\rho, U\sigma U^\dagger))^2&=&\sum_{i=1}^n\lambda_i(\rho)\lambda_{n-i+1}(\sigma),
\end{eqnarray*}
which is the lower bound of inequality (\ref{eq:rearr}), with $x_i=\lambda_i(\rho)$ and $y_i=\lambda_i(\sigma)$.
All other permutation matrices $U$ just yield something in between these two bounds, thus giving inequality (\ref{eq:rearr}).

\subsection{The Spectral Geometric Mean}
 Let $A\# B=A^{1/2}(A^{-1/2}BA^{-1/2})^{1/2}A^{1/2}$ be the geometric mean between positive semidefinite matrices $A$ and $B$. In \cite{UC09}, the authors show that for bipartite states $\rho$ and $\sigma$, the fidelity of Alice's reduced states $\rho^A$ and  $\sigma^A$ is related to the geometric mean of Bob's reduced states  $\rho^B$ and  $\sigma^B$:
 \[
 F(\rho^A, \sigma^A)=\tr(\rho^B\# \sigma^B).
 \]
 Here we show that for general states $\rho$ and $\sigma$ (not necessarily bipartite), their fidelity is intimately related to their \emph{spectral geometric mean}.

 The spectral geometric mean between positive semidefinite matrices $A$ and $B$ is given by $A\diamond  B =(A^{-1}\# B)^{1/2}A(A^{-1}\# B)^{1/2}$ \cite{FP97}, which has the useful feature that $(A\diamond B)^2 $ is similar to $AB$.

 We have, for  positive semidefinite matrices $A$ and $B$,
 \begin{eqnarray*}
 \tr(A\diamond B)&=&\tr((A^{-1/2}(A^{1/2}BA^{1/2})^{1/2}A^{-1/2})^{1/2}A(A^{-1/2}(A^{1/2}BA^{1/2})^{1/2}A^{-1/2})^{1/2})\\
 &=&\tr(A^{-1/2}(A^{1/2}BA^{1/2})^{1/2}A^{-1/2}A)\\
 &=&\tr((A^{1/2}BA^{1/2})^{1/2}).
 \end{eqnarray*}
 In particular, for density matrices $\rho$ and $\sigma$, we have
 \begin{eqnarray}
 F(\rho, \sigma)=\tr(\rho\diamond \sigma)
 \end{eqnarray}
 with $\rho\diamond \sigma$ similar to $\sqrt{\rho\sigma}$ (since  $(\rho\diamond \sigma)^2$ is similar to ${\rho\sigma}$).
 Re-writing the trace as a sum of eigenvalues, we have
 \begin{eqnarray}
  F(\rho, \sigma)=\sum_i\lambda_i(\rho\diamond\sigma)=\sum_i\lambda_i(\sqrt{\rho\sigma}).
 \end{eqnarray}

 \subsection{Maximum Fidelity}
 The \emph{maximum output fidelity} of two channels $\Phi, \Psi$ is define as \cite{Ros09, KW00}
 \[
 F_{\max}(\Phi, \Psi)=\max_{\rho, \sigma}F(\Psi(\rho), \Phi(\sigma)),
 \]
 where the maximum is taken over all density matrices $\rho$ and $\sigma$. This maximum fidelity is connected to the diamond norm $\|\cdot\|_\diamond$, the dual of the completely bounded norm, via the following lemma
 \begin{lemma}\cite{KW00}
Let $\Phi,\Psi:\mathfrak{B}(H)\rightarrow\mathfrak{B}(K) $  be quantum channels  with Stinespring dilations
\begin{eqnarray*}
\Phi(X)&=&\tr_BUXU^\dagger\\
\Psi(X)&=&\tr_BVXV^\dagger,
\end{eqnarray*}
where $U,V:H\rightarrow B\otimes K$ are unitaries.
Let $\Gamma$ be the linear map given by $\Gamma(X)=\tr_KUXV^\dagger$. Then $F_{\max}(\Phi, \Psi)=\|\Gamma\|_\diamond$.
 \end{lemma}
Taking the maximum allows one to interpret fidelity in terms of the diamond norm. This interpretation has been used in \cite{Ros09, KW00} with respect to quantum interactive proof systems.

However, $F_{\max}(\Phi, \Psi)=1$ whenever the ranges of the two channels overlap. Thus, we propose $F_{\min}(\Phi, \Psi)$ as a more informative measure of distance between two channels, in the sense that it will only give 1 when the channels are equal, allowing for more useful comparisons between channels.

\section*{Acknowledgements}
R.P.\ was supported by NSERC Discovery Grant number 400550. S.P.\ was supported by NSERC Discovery Grant number 1174582. R.P.\ and S.P.\  wish to acknowledge the Sanya International Mathematics Forum, which hosted the twelveth Workshop on Numerical Ranges and Numerical Radii, where this work was initiated.  The authors would like to thank Dr.~Lin Zhang for many helpful suggestions about references.

\end{document}